\documentclass{lmcs}
\usepackage{amssymb}
\usepackage{tikz}
\usepackage[normalem]{ulem}  
\usepackage[linguistics]{forest}
\usepackage[numbers,sort&compress]{natbib}  
\usepackage[hidelinks]{hyperref}  
\usepackage{amssymb}
\usepackage[hidelinks]{hyperref}  
\usepackage[capitalize,noabbrev]{cleveref}

\newcommand{\SquareDiamond}{%
	\tikz[baseline=-0.55ex,scale=0.09,line width=0.5pt]
	\draw[rotate=45] (-1,0) rectangle (1,2);
}

\begin{document}
	
	\title{
		Knowledge Reasoning Involving Four Types of Syllogisms}
	\author[Long Wei]{Long Wei}[a]
	\address{Department of Philosophy (Zhuhai), Sun Yat-Sen University, Zhuhai, China}
	\email{657703460@qq.com}
	
	\author[Liheng Hao]{Liheng Hao}[b]
	\address{School of Engineering and Materials Science, Queen Mary University of London, 
		London, United Kingdom}
	\email{haolihengxtw@163.com (Corresponding author)}
	\thanks{Supported by the National Social Science Fund of China under Grant No.24FZXB068.}
	
	\begin{abstract}
		This paper studies the validity and discourse reasoning of non-trivial generalized syllogisms involving the quantifiers in Square\{\textit{most}\} and Square\{\textit{all}\} from the perspective of knowledge reasoning. Firstly, this paper presents knowledge representations for these syllogisms and formally proves the validity of generalized syllogism AMI-1. Subsequently, 19 non-trivial generalized syllogisms, 22 non-trivial valid generalized modal syllogisms, 8 valid classical syllogisms, and 24 valid classical modal syllogisms are respectively deduced from the valid generalized syllogism AMI-1 on the basis of deductive reasoning. Additionally, this paper discusses how to judge the validity of discourse reasoning nested by the above four types of syllogisms, which have four types of figures and different forms. In conclusion, such formal deductions not only provide a theoretical foundation for English language information processing, but also provide methodological insights for studying other syllogistic systems.
	\end{abstract}
	

	\maketitle
	
	\section{Introduction}
	
	In natural language reasoning, there are the following four common syllogisms studied in this paper: classical syllogisms \cite{Wes89,Mos08,Hao23}; classical modal ones \cite{McC63,Mal13,Zha23,ZW25}; generalized ones \cite{MN12,EM15,CL24}; and generalized modal ones \cite{XZ23,YZ24}. These four types of syllogisms are classified into four types according to the kinds of quantifiers and modalities they involve. Specifically, classical syllogisms only contain the following four classical quantifiers: all, some, no, and not all, which are regarded as trivial generalized quantifiers. A non-trivial generalized syllogism contains at least one non-trivial generalized quantifier, such as two thirds, several, infinitely many, most. A modal syllogism contains at least one and at most three modalities (the necessary one $\Box$ or the possible one $\SquareDiamond$). A non-trivial classical modal syllogism contains only classical quantifiers and modalities, while a non-trivial generalized modal syllogism contains at least one non-trivial generalized quantifier and modality. In order to ensure logical consistency, each type of valid syllogisms follows specific rules.
	
	Syllogistic reasoning is a common mode of inference in natural language and an important subject of study in the field of knowledge reasoning. This paper focuses exclusively on the four types of syllogisms discussed above. Previous studies on syllogisms have mainly concentrated on questions of validity. In this paper, the generalized syllogism AMI-1 serves as the foundational one, from which other types of syllogisms are derived. Furthermore, the reducible relationships between/among the four types of syllogisms are systematically analyzed.
	
	There are countless generalized quantifiers in natural language, and the generalized syllogisms studied in this paper only involve the following common quantifiers: the trivial generalized quantifiers from Square\{\textit{all}\}=\{\textit{no, all, not all, some}\} and the non-trivial generalized ones from Square\{\textit{most}\}=\{\textit{most, fewer than half of the, at least half of the, at most half of the}\}. Focusing on the valid generalized syllogism AMI-1, this paper first discusses how to infer more generalized syllogisms from the generalized syllogism AMI-1, and then demonstrates how to deduce the other three types of syllogisms from this generalized syllogism.
	\section{Preliminaries}
	In this paper, let \textit{n, p, w, c, k, z, r} be lexical variables, and \textit{D} be their domain. The sets composed of \textit{n, p, w, c, k, z, r} are \textit{N, P, W, C, K, Z, R}, respectively. Let $\lambda$, $\theta$, $\mu$, and $\pi$ be \textit{wff} (i.e. well-formed formula). Let $Q$ be a quantifier, $Q\neg$, $\neg Q$, and $\neg Q\neg$ be its inner, outer, and dual negative one, respectively. These four quantifiers form a modern square, denoted as $\mathrm{Square}\{Q\}$. ‘$|P \cap W|$’ states the cardinality of the intersection of the set $P$ and $W$. 
	‘$\vdash \lambda$’ means that the \textit{wff} $\lambda$ is provable, and 
	‘$\lambda =_{\mathrm{def}} \theta$’ that $\lambda$ can be defined by $\theta$. 
	The connectives (such as $\neg$, $\rightarrow$, $\land$, $\leftrightarrow$) are 
	common operators in classical logic \cite{Hal74}. 
	
	The syllogisms studied in this work only involve the 8 quantifiers in 
	Square\{\textit{all}\}=\{\textit{no, all, not all, some}\}
	and Square\{\textit{most}\}=\{\textit{fewer than half of the, most, at least half of the, at most half of the}\}. Therefore, we only consider the following eight types of propositions: 
	$\textit{no}(p,w)$, $\textit{all}(p,w)$, $\textit{not all}(p,w)$, $\textit{some}(p,w)$, 
	$\textit{fewer than half of the}(p,w)$, $\textit{most}(p,w)$, 
	$\textit{at least half of the}(p,w)$, $\textit{at most half of the}(p,w)$, 
	and they are denoted as: Proposition $E$, $A$, $O$, $I$, $F$, $M$, $S$, and $H$, respectively. 
	A non-trivial generalized syllogism studied in this paper contains at least one quantifier from 
	$\mathrm{Square}\{\textit{most}\}$. 
	For instance, the first figure syllogism 
	$\textit{all}(n,w) \land \textit{most}(p,n) \rightarrow \textit{some}(p,w)$ 
	is shortened to AMI-1. The following is an instance of this syllogism:
	\begin{exa}\leavevmode\par
	\noindent Major premise: All smartphones have touchscreens.\\
	\noindent Major premise: Most mobile phones sold this year in this shop are smartphones.\\
	\noindent Conclusion: Therefore, some mobile phones sold this year in this shop have touchscreens.
		
	\medskip
		
		Let $n$, $w$, and $p$ be the variables that represent 
		`\textit{smartphones}', `\textit{mobile phones that have touchscreens}', 
		and `\textit{mobile phones sold this year in this shop}'. 
		Then, this instance can be symbolized as 
		$\textit{all}(n,w)\land\textit{most}(p,n)\rightarrow\textit{some}(p,w)$, 
		which is abbreviated as AMI-1. 
		Other formal representations are similar to this.
	\end{exa}
	
	\section{A Fragment of Generalized Syllogism System}
	
	This fragment consists of the following six parts:
	\subsection{Primitive Symbols}
	\begin{itemize}
		\item lexical variables: $n, p, w$
		\item quantifiers: \textit{all}, \textit{most}
		\item operators: $\land, \lnot$
		\item brackets: $(, )$
	\end{itemize}
	
	\subsection{Formation Rules}
	A \textit{well-formed formula} (abbreviated as \textit{wff}) is an expression involving quantifiers, lexical variables, and connectives, which can be formed using the following rules:
	\begin{itemize}
		\item If $Q$ is a quantifier, $p$ and $w$ are variables, then $Q(p, w)$ is a \textit{wff}.
		\item If $\lambda$ is a \textit{wff}, then so is $\lnot \lambda$.
		\item If $\lambda$ and $\theta$ are \textit{wffs}, then so is $\lambda \rightarrow \theta$.
	\end{itemize}
	
	\subsection{Basic Axioms}
	\begin{enumerate}[label={A\arabic*:}]
		\item If $\lambda$ is a valid wff in classical logic, then $\vdash \lambda$.
		\item $\vdash \textit{all}(n, w) \land \textit{most}(p, n) \rightarrow \textit{some}(p, w)$ (i.e., the syllogism AMI-1).
	\end{enumerate}
	
	\subsection{Deductive Rules}
	\begin{enumerate}[label={Rule \arabic*:}]
		\item (\textit{Antecedent strengthening}): From $\vdash (\lambda \land \theta \rightarrow \mu)$ and $\vdash (\pi \rightarrow \lambda)$ infer $\vdash (\pi \land \theta \rightarrow \mu)$.
		\item (\textit{Subsequent weakening}): From $\vdash (\lambda \land \theta \rightarrow \mu)$ and $\vdash (\mu \rightarrow \pi)$ infer $\vdash (\lambda \land \theta \rightarrow \pi)$.
		\item (\textit{Anti-syllogism}): From $\vdash (\lambda \land \theta \rightarrow \mu)$ infer $\vdash (\lnot \mu \land \lambda \rightarrow \lnot \theta)$.
	\end{enumerate}

	\subsection{Related Definitions}
	\begin{enumerate}[label={D\arabic*:}]
		\item $(\lambda \land \theta) =_{\mathrm{def}} \lnot (\lambda \rightarrow \lnot \theta)$;
		\item $(\lambda \leftrightarrow \theta) =_{\mathrm{def}} (\lambda \rightarrow \theta) \land (\theta \rightarrow \lambda)$;
		\item $(Q \lnot)(p, w) =_{\mathrm{def}} Q(p, D{-}w)$;
		\item $(\lnot Q)(p, w) =_{\mathrm{def}}$ It is not the case that $Q(p, w)$;
		\item $\textit{all}(p, w)$ is true iff $P \subseteq W$ in any real world;
		\item $\textit{some}(p, w)$ is true iff $P \cap W \neq \emptyset$ in any real world;
		\item $\textit{most}(p, w)$ is true iff $|P \cap W| > 0.5|P|$ in any real world;
		\item $\SquareDiamond \textit{some}(p, w)$ is true iff $P \cap W \neq \emptyset$ in at least one possible world;
		\item $\SquareDiamond \textit{most}(p, w)$ is true iff $|P \cap W| > 0.5|P|$ in at least one possible world.
	\end{enumerate}
	
	\subsection{Relevant Facts\\\hspace*{1.8em}Fact 1 (Inner Negation)}
	
	\begin{enumerate}[label={[1.\arabic*]}]
		\item $\vdash \textit{all}(p, w) \leftrightarrow \textit{no}\lnot(p, w);$
		\item $\vdash \textit{no}(p, w) \leftrightarrow \textit{all}\lnot(p, w);$
		\item $\vdash \textit{some}(p, w) \leftrightarrow \textit{not all}\lnot(p, w);$
		\item $\vdash \textit{not all}(p, w) \leftrightarrow \textit{some}\lnot(p, w);$
		\item $\vdash \textit{most}(p, w) \leftrightarrow \textit{fewer than half of the}\lnot(p, w);$
		\item $\vdash \textit{fewer than half of the}(p, w) \leftrightarrow \textit{most}\lnot(p, w);$
		\item $\vdash \textit{at least half of the}(p, w) \leftrightarrow \textit{at most half of the}\lnot(p, w);$
		\item $\vdash \textit{at most half of the}(p, w) \leftrightarrow \textit{at least half of the}\lnot(p, w).$
	\end{enumerate}
	
	\medskip
	
	\textbf{Fact 2 (Outer Negation):}
	\begin{enumerate}[label={[2.\arabic*]}]
		\item $\vdash \lnot \textit{all}(p, w) \leftrightarrow \textit{not all}(p, w);$
		\item $\vdash \lnot\textit{not all}(p, w) \leftrightarrow \textit{all}(p, w);$
		\item $\vdash \lnot\textit{no}(p, w) \leftrightarrow \textit{some}(p, w);$
		\item $\vdash \lnot\textit{some}(p, w) \leftrightarrow \textit{no}(p, w);$
		\item $\vdash \lnot\textit{most}(p, w) \leftrightarrow \textit{at most half of the}(p, w);$
		\item $\vdash \lnot\textit{at most half of the}(p, w) \leftrightarrow \textit{most}(p, w);$
		\item $\vdash \lnot\textit{fewer than half of the}(p, w) \leftrightarrow \textit{at least half of the}(p, w);$
		\item $\vdash \lnot\textit{at least half of the}(p, w) \leftrightarrow \textit{fewer than half of the}(p, w).$
	\end{enumerate}
	
	\medskip
	
	\textbf{Fact 3 (Symmetry):}
	\begin{enumerate}[label={[3.\arabic*]}]
		\item $\vdash \textit{some}(p, w) \leftrightarrow \textit{some}(w, p);$
		\item $\vdash \textit{no}(p, w) \leftrightarrow \textit{no}(w, p).$
	\end{enumerate}
	
	\medskip
	
	\textbf{Fact 4 (Subordination):}
	\begin{enumerate}[label={[4.\arabic*]}]
		\item $\vdash \textit{all}(p, w) \rightarrow \textit{some}(p, w);$
		\item $\vdash \textit{no}(p, w) \rightarrow \textit{not all}(p, w);$
		\item $\vdash \textit{all}(p, w) \rightarrow \textit{most}(p, w);$
		\item $\vdash \textit{most}(p, w) \rightarrow \textit{some}(p, w);$
		\item $\vdash \textit{at least half of the}(p, w) \rightarrow \textit{some}(p, w);$
		\item $\vdash \textit{all}(p, w) \rightarrow \textit{at least half of the}(p, w);$
		\item $\vdash \textit{at most half of the}(p, w) \rightarrow \textit{not all}(p, w);$
		\item $\vdash \textit{fewer than half of the}(p, w) \rightarrow \textit{not all}(p, w);$
		\item $\vdash \textit{at least half of the}(p, w) \rightarrow \textit{most}(p, w);$
		\item $\vdash \Box Q(p, w) \rightarrow Q(p, w);$
		\item $\vdash \Box Q(p, w) \rightarrow \SquareDiamond Q(p, w);$
		\item $\vdash Q(p, w) \rightarrow \SquareDiamond Q(p, w).$
	\end{enumerate}
	
	\medskip
	
	\textbf{Fact 5 (Duality):}
	\begin{enumerate}[label={[5.\arabic*]}]
		\item $\vdash \lnot\Box Q(p, w) \leftrightarrow \SquareDiamond \lnot Q(p, w);$
		\item $\vdash \lnot\SquareDiamond Q(p, w) \leftrightarrow \Box \lnot Q(p, w).$
	\end{enumerate}
	
	\medskip
	
	Fact 1--5 are the theoretical basis of first-order logic and generalized quantifier theory~\cite{PW06}, so their proofs are omitted.

	\section{Knowledge Mining about the Generalized Syllogism AMI-1}
	A valid syllogism is considered reducible if the logical validity of other syllogisms can be formally inferred from it. From the perspective of knowledge mining, it means that a piece of useful knowledge can be used to mine another piece of useful knowledge. As formally demonstrated in Theorem~\ref{thm:4.2}, this reducible principle specifically manifests in the following manner: The validity of generalized syllogisms positioned after the implication symbol (i.e.$\rightarrow$) can be deduced from the validity of the foundational syllogism AMI-1. This establishes reducible relationships between/among these logically valid syllogistic forms.

	\begin{thm}[AMI-1]
		The generalized syllogism $\textit{all}(n, w) \land \textit{most}(p, n) \rightarrow \textit{some}(p, w)$ is valid.
	\end{thm}
	
	\begin{proof}
		Suppose that the two premises $\textit{all}(n, w)$ and $\textit{most}(p, n)$ are true. 
		According to Definitions D5 and D7, $N \subseteq W$ and $|P \cap N| > 0.5|P|$ are true in any real world, thus $|P \cap W| > 0.5|P|$. 
		In line with Definition D6, $\textit{some}(p, w)$ is true when and only when $P \cap W \neq \emptyset$ is true in any real world. 
		Since $|P \cap W| > 0.5|P|$, $\textit{some}(p, w)$ is true. 
		That means that the syllogism AMI-1 is valid.
	\end{proof}
	
	\begin{thm}\label{thm:4.2}
		There are at least the following 19 valid non-trivial generalized syllogisms that can be mined from the syllogism AMI-1:
	\end{thm}
	
	\begin{enumerate}[label={(\arabic*)}, leftmargin=3em, itemsep=3pt, topsep=4pt]
		\item $\vdash \text{AMI-1} \rightarrow \text{MAI-4};$
		\item $\vdash \text{AMI-1} \rightarrow \text{AEH-2};$
		\item $\vdash \text{AMI-1} \rightarrow \text{AEH-2} \rightarrow \text{AEH-4};$
		\item $\vdash \text{AMI-1} \rightarrow \text{EMO-3};$
		\item $\vdash \text{AMI-1} \rightarrow \text{EMO-3} \rightarrow \text{EMO-4};$
		\item $\vdash \text{AMI-1} \rightarrow \text{EMO-1};$
		\item $\vdash \text{AMI-1} \rightarrow \text{EMO-1} \rightarrow \text{EMO-2};$
		\item $\vdash \text{AMI-1} \rightarrow \text{EMO-3} \rightarrow \text{AMI-3};$
		\item $\vdash \text{AMI-1} \rightarrow \text{EMO-3} \rightarrow \text{AMI-3} \rightarrow \text{MAI-3};$
		\item $\vdash \text{AMI-1} \rightarrow \text{AEH-2} \rightarrow \text{EAH-2};$
		\item $\vdash \text{AMI-1} \rightarrow \text{AEH-2} \rightarrow \text{EAH-2} \rightarrow \text{EAH-1};$
		\item $\vdash \text{AMI-1} \rightarrow \text{ASI-1};$
		\item $\vdash \text{AMI-1} \rightarrow \text{ASI-1} \rightarrow \text{SAI-4};$
		\item $\vdash \text{AMI-1} \rightarrow \text{ASI-1} \rightarrow \text{ESO-3};$
		\item $\vdash \text{AMI-1} \rightarrow \text{ASI-1} \rightarrow \text{ESO-3} \rightarrow \text{ESO-4};$
		\item $\vdash \text{AMI-1} \rightarrow \text{ASI-1} \rightarrow \text{AEF-2};$
		\item $\vdash \text{AMI-1} \rightarrow \text{ASI-1} \rightarrow \text{AEF-2} \rightarrow \text{AEF-4};$
		\item $\vdash \text{AMI-1} \rightarrow \text{EMO-3} \rightarrow \text{AMI-3} \rightarrow \text{ASI-3};$
		\item $\vdash \text{AMI-1} \rightarrow \text{EMO-3} \rightarrow \text{AMI-3} \rightarrow \text{ASI-3} \rightarrow \text{SAI-3}.$
	\end{enumerate}
	
	\begin{proof}\leavevmode\par
		
		\noindent
		\makebox[0pt][l]{[1] $\vdash all(n,w) \land most(p,n) \rightarrow some(p,w)$}%
		\hfill (i.e. AMI-1, Axiom A2)\par
		
		\noindent
		\makebox[0pt][l]{[2] $\vdash all(n,w) \land most(p,n) \rightarrow some(w,p)$}%
		\hfill (i.e. MAI-4, by [1] and Fact [3.1])\par
		
		\noindent
		\makebox[0pt][l]{[3] $\vdash \lnot some(p,w) \land all(n,w) \rightarrow \lnot most(p,n)$}%
		\hfill (by [1] and Rule 3)\par
		
		\noindent
		\makebox[0pt][l]{[4] $\vdash no(p, w) \land all(n, w) \rightarrow at~most~half~of~the(p, n)$} \par
		\hfill (i.e. AEH-2, by [3], Fact [2.4] and [2.5])\par
		
		\noindent
		\makebox[0pt][l]{[5] $\vdash no(w, p) \land all(n, w) \rightarrow at~most~half~of~the(p, n)$}%
		\hfill (i.e. AEH-4, by [4] and Fact [3.2])\par
		
		\noindent
		\makebox[0pt][l]{[6] $\vdash \lnot some(p, w) \land most(p, n) \rightarrow \lnot all(n, w)$}%
		\hfill (by [1] and Rule 3)\par
		
		\noindent
		\makebox[0pt][l]{[7] $\vdash no(p, w) \land most(p, n) \rightarrow not~all(n, w)$}%
		\hfill (i.e. EMO-3, by [6], Fact [2.1] and [2.4])\par
		
		\noindent
		\makebox[0pt][l]{[8] $\vdash no(w, p) \land most(p, n) \rightarrow not~all(n, w)$}%
		\hfill (i.e. EMO-4, by [7] and Fact [3.2])\par
		
		\noindent
		\makebox[0pt][l]{[9] $\vdash no\lnot(n, w) \land most(p, n) \rightarrow not~all\lnot(p, w)$}%
		\hfill (by [1], Fact [1.1] and [1.3])\par
		
		\noindent
		\makebox[0pt][l]{[10] $\vdash no(n, D{-}ww) \land most(p, n) \rightarrow not~all(p, D{-}w)$}\par
		\hfill (i.e. EMO-1, by [9] and Definition D3)\par
		
		\noindent
		\makebox[0pt][l]{[11] $\vdash no(D{-}w, n) \land most(p, n) \rightarrow not~all(p, D{-}w)$}%
		\hfill (i.e. EMO-2, by [10] and Fact [3.2])\par
		
		\noindent
		\makebox[0pt][l]{[12] $\vdash all\lnot(p, w) \land most(p, n) \rightarrow some\lnot(n, w)$}%
		\hfill (by [7], Fact [1.2] and [1.4])\par
		
		\noindent
		\makebox[0pt][l]{[13] $\vdash all(p, D{-}w) \land most(p, n) \rightarrow some(n, D{-}w)$}%
		\hfill (i.e. AMI-3, by [12] and Definition D3)\par
		
		\noindent
		\makebox[0pt][l]{[14] $\vdash all(p, D{-}w) \land most(p, n) \rightarrow some(D{-}w, n)$}%
		\hfill (i.e. MAI-3, by [13] and Fact [3.1])\par
		
		\noindent
		\makebox[0pt][l]{[15] $\vdash all\lnot(p, w) \land no\lnot(n, w) \rightarrow at~most~half~of~the(p, n)$}%
		\hfill (by [4], Fact [1.1] and [1.2])\par
		
		\noindent
		\makebox[0pt][l]{[16] $\vdash all(p, D{-}w) \land no(n, D{-}w) \rightarrow at~most~half~of~the(p, n)$}\par
		\hfill (i.e. EAH-2, by [15] and Definition D3)\par
		
		\noindent
		\makebox[0pt][l]{[17] $\vdash all(p, D{-}w) \land no(D{-}w, n) \rightarrow at~most~half~of~the(p, n)$}\par
		\hfill (i.e. EAH-1, by [16] and Fact [3.2])\par
		
		\noindent
		\makebox[0pt][l]{[18] $\vdash all(n, w) \land at~least~half~of~the(p, n) \rightarrow some(p, w)$}\par
		\hfill (i.e. ASI-1, by [1], Fact [4.5] and Rule 1)\par
		
		\noindent
		\makebox[0pt][l]{[19] $\vdash all(n, w) \land at~least~half~of~the(p, n) \rightarrow some(w, p)$}\par
		\hfill (i.e. SAI-4, by [18] and Fact [3.1])\par
		
		\noindent
		\makebox[0pt][l]{[20] $\vdash \lnot some(p, w) \land at~least~half~of~the(p, n) \rightarrow \lnot all(n, w)$}%
		\hfill (by [18] and Rule 3)\par
		
		\noindent
		\makebox[0pt][l]{[21] $\vdash no(p, w) \land at~least~half~of~the(p, n) \rightarrow not~all(n, w)$}\par
		\hfill (i.e. ESO-3, by [20], Fact [2.1] and [2.4])\par
		
		\noindent
		\makebox[0pt][l]{[22] $\vdash no(w, p) \land at~least~half~of~the(p, n) \rightarrow not~all(n, w)$}\par
		\hfill (i.e. ESO-4, by [21] and Fact [3.2])\par
		
		\noindent
		\makebox[0pt][l]{[23] $\vdash \lnot some(p, w) \land all(n, w) \rightarrow \lnot at~least~half~of~the(p, n)$}%
		\hfill (by [18] and Rule 3)\par
		
		\noindent
		\makebox[0pt][l]{[24] $\vdash no(p, w) \land all(n, w) \rightarrow fewer~than~half~of~the(p, n)$}\par
		\hfill (i.e. AEF-2, by [23], Fact [2.4] and [2.8])\par
		
		\noindent
		\makebox[0pt][l]{[25] $\vdash no(w, p) \land all(n, w) \rightarrow fewer~than~half~of~the(p, n)$}\par
		\hfill (i.e. AEF-4, by [24] and Fact [3.2])\par
		
		\noindent
		\makebox[0pt][l]{[26] $\vdash all(p, D{-}w) \land at~least~half~of~the(p, n) \rightarrow some(n, D{-}w)$}\par
		\hfill (i.e. ASI-3, by [13], Fact [4.9] and Rule 1)\par
		
		\noindent
		\makebox[0pt][l]{[27] $\vdash all(p, D{-}w) \land at~least~half~of~the(p, n) \rightarrow some(D{-}w, n)$}\par
		\hfill (i.e. SAI-3, by [26] and Fact [3.1])\par

	\end{proof}
	Theorem~\ref{thm:4.3} discusses the method for deducing valid generalized modal syllogisms from the valid generalized syllogism AMI-1, and the following steps should be taken: (1) The resulting syllogism AM$\SquareDiamond$I-1 is deduced from AMI-1 by the subsequent weakening rule. (2) The resulting syllogism $\Box$AMI-1 is deduced from AMI-1 by the antecedent strengthening rule. (3) Other generalized modal syllogisms can be deduced from AM$\SquareDiamond$I-1 or $\Box$AMI-1 by the application of deductive reasoning.
	
	\begin{thm}\label{thm:4.3}
		There are at least the following 22 valid non-trivial generalized modal syllogisms that can be mined from the generalized syllogism AMI-1:
	\end{thm}
	\begin{enumerate}[label={(\arabic*)}, leftmargin=3em, itemsep=3pt, topsep=4pt]
		\item $\vdash \text{AMI-1} \rightarrow \text{AM$\SquareDiamond$I-1};$  
		\item $\vdash \text{AMI-1} \rightarrow \text{AM$\SquareDiamond$I-1} \rightarrow \Box\text{AM$\SquareDiamond$I-1};$
		\item $\vdash \text{AMI-1} \rightarrow \text{A}\text{M$\SquareDiamond$I-1} \rightarrow \text{A}\Box\text{M$\SquareDiamond$I-1};$
		\item $\vdash \text{AMI-1} \rightarrow \text{AM$\SquareDiamond$I-1} \rightarrow \text{A}\Box\text{M$\SquareDiamond$I-1} \rightarrow \Box\text{A$\Box$M$\SquareDiamond$I-1};$
		\item $\vdash \text{AMI-1} \rightarrow \text{AM$\SquareDiamond$I-1} \rightarrow \text{MA$\SquareDiamond$I-4};$
		\item $\vdash \text{AMI-1} \rightarrow \text{AM$\SquareDiamond$I-1} \rightarrow \Box\text{AM$\SquareDiamond$I-1} \rightarrow \text{M$\Box$A$\SquareDiamond$I-4};$
		\item $\vdash \text{AMI-1} \rightarrow \text{AM$\SquareDiamond$I-1} \rightarrow \text{A}\Box\text{M$\SquareDiamond$I-1} \rightarrow \text{MA$\Box$I-4};$
		\item $\vdash \text{AMI-1} \rightarrow \text{AM$\SquareDiamond$I-1} \rightarrow \text{A}\Box\text{M$\SquareDiamond$I-1} \rightarrow \Box\text{A$\Box$M$\SquareDiamond$I-1} \rightarrow \Box\text{M$\Box$A$\SquareDiamond$I-4};$
		\item $\vdash \text{AMI-1} \rightarrow \Box\text{AMI-1};$
		\item $\vdash \text{AMI-1} \rightarrow \Box\text{AMI-1} \rightarrow \Box\text{A$\Box$MI-1};$
		\item $\vdash \text{AMI-1} \rightarrow \Box\text{AMI-1} \rightarrow \text{M$\Box$AI-4};$
		\item $\vdash \text{AMI-1} \rightarrow \Box\text{AMI-1} \rightarrow \Box\text{A$\Box$MI-1} \rightarrow \Box\text{M$\Box$AI-4};$
		\item $\vdash \text{AMI-1} \rightarrow \text{AM$\SquareDiamond$I-1} \rightarrow \Box\text{EMO-3};$
		\item $\vdash \text{AMI-1} \rightarrow \text{AM$\SquareDiamond$I-1} \rightarrow \Box\text{EMO-3} \rightarrow \Box\text{E$\Box$MO-3};$
		\item $\vdash \text{AMI-1} \rightarrow \text{AM$\SquareDiamond$I-1} \rightarrow \Box\text{EMO-3} \rightarrow \Box\text{EM$\SquareDiamond$O-3};$
		\item $\vdash \text{AMI-1} \rightarrow \text{AM$\SquareDiamond$I-1} \rightarrow \Box\text{EMO-3} \rightarrow \Box\text{E$\Box$MO-3} \rightarrow \Box\text{E$\Box$M$\SquareDiamond$O-3};$
		\item $\vdash \text{AMI-1} \rightarrow \text{AM$\SquareDiamond$I-1} \rightarrow \Box\text{EMO-3} \rightarrow \Box\text{EMO-4};$
		\item $\vdash \text{AMI-1} \rightarrow \text{AM$\SquareDiamond$I-1} \rightarrow \Box\text{EMO-3} \rightarrow \Box\text{E$\Box$MO-3} \rightarrow \Box\text{E$\Box$MO-4};$
		\item $\vdash \text{AMI-1} \rightarrow \text{AM$\SquareDiamond$I-1} \rightarrow \Box\text{EMO-3} \rightarrow \Box\text{EM$\SquareDiamond$O-3} \rightarrow \Box\text{EM$\SquareDiamond$O-4};$
		\item $\vdash \text{AMI-1} \rightarrow \text{AM$\SquareDiamond$I-1} \rightarrow \Box\text{EMO-3} \rightarrow \Box\text{E$\Box$MO-3} \rightarrow \Box\text{E$\Box$M$\SquareDiamond$O-3} \rightarrow \Box\text{E$\Box$M$\SquareDiamond$O-4};$
		\item $\vdash \text{AMI-1} \rightarrow \Box\text{AMI-1} \rightarrow \text{EM$\SquareDiamond$O-3};$
		\item $\vdash \text{AMI-1} \rightarrow \Box\text{AMI-1} \rightarrow \text{EM$\SquareDiamond$O-3} \rightarrow \text{EM$\SquareDiamond$O-4}.$
	\end{enumerate}
	
	\begin{proof}\leavevmode\par
		\noindent
		\makebox[0pt][l]{[1] $\vdash all(n,w)\land most(p,n)\rightarrow some(p,w)$}%
		\hfill (i.e. AMI-1, Axiom A2)\par
		
		\noindent
		\makebox[0pt][l]{[2] $\vdash all(n,w)\land most(p,n)\rightarrow\SquareDiamond some(p,w)$}%
		\hfill (i.e. AM$\SquareDiamond$I-1, by [1], Fact [4.12] and Rule 2)\par
		
		\noindent
		\makebox[0pt][l]{[3] $\vdash \Box all(n,w)\land most(p,n)\rightarrow\SquareDiamond some(p,w)$}%
		\hfill (i.e. $\Box$AM$\SquareDiamond$I-1, by [2], Fact [4.10] and Rule 1)\par
		
		\noindent
		\makebox[0pt][l]{[4] $\vdash all(n,w)\land \Box most(p,n)\rightarrow\SquareDiamond some(p,w)$}%
		\hfill (i.e. A$\Box$M$\SquareDiamond$I-1, by [2], Fact [4.10] and Rule 1)\par
		
		\noindent
		\makebox[0pt][l]{[5] $\vdash \Box all(n,w)\land \Box most(p,n)\rightarrow\SquareDiamond some(p,w)$}\par
		\hfill (i.e. $\Box$A$\Box$M$\SquareDiamond$I-1, by [4], Fact [4.10] and Rule 1)\par
		
		\noindent
		\makebox[0pt][l]{[6] $\vdash all(n,w)\land most(p,n)\rightarrow some(w,p)$}%
		\hfill (i.e. MAI-4, by [2] and Fact [3.1])\par
		
		\noindent
		\makebox[0pt][l]{[7] $\vdash \Box all(n,w)\land most(p,n)\rightarrow\SquareDiamond some(w,p)$}%
		\hfill (i.e. M$\Box$A$\SquareDiamond$I-4, by [3] and Fact [3.1])\par
		
		\noindent
		\makebox[0pt][l]{[8] $\vdash all(n,w)\land \Box most(p,n)\rightarrow\SquareDiamond some(w,p)$}%
		\hfill (i.e. MA$\Box\SquareDiamond$I-4, by [4] and Fact [3.1])\par
		
		\noindent
		\makebox[0pt][l]{[9] $\vdash all(n,w)\land \Box most(p,n)\rightarrow\SquareDiamond some(w,p)$}%
		\hfill (i.e. M$\Box$A$\SquareDiamond$I-4, by [5] and Fact [3.1])\par
		
		\noindent
		\makebox[0pt][l]{[10] $\vdash \Box all(n,w)\land most(p,n)\rightarrow some(p,w)$}%
		\hfill (i.e. $\Box$AMI-1, by [1], Fact [4.10] and Rule 1)\par
		
		\noindent
		\makebox[0pt][l]{[11] $\vdash \Box all(n,w)\land most(p,n)\rightarrow some(p,w)$}\par
		\hfill (i.e. $\Box$A$\Box$MI-1, by [10], Fact [4.10] and Rule 1)\par
		
		\noindent
		\makebox[0pt][l]{[12] $\vdash \Box all(n,w)\land most(p,n)\rightarrow some(w,p)$}%
		\hfill (i.e. M$\Box$AI-4, by [10] and Fact [3.1])\par
		
		\noindent
		\makebox[0pt][l]{[13] $\vdash \Box all(n,w)\land \Box most(p,n)\rightarrow some(w,p)$}%
		\hfill (i.e. $\Box$M$\Box$AI-4, by [11] and Fact [3.1])\par
		
		\noindent
		\makebox[0pt][l]{[14] $\vdash \lnot \SquareDiamond some(p,w)\land most(p,n)\rightarrow \lnot all(n,w)$}%
		\hfill (by [2] and Rule 3)\par
		
		\noindent
		\makebox[0pt][l]{[15] $\vdash \Box \lnot some(p,w)\land most(p,n)\rightarrow not~all(n,w)$}%
		\hfill (by [14], Fact [2.1] and [5.2])\par
		
		\noindent
		\makebox[0pt][l]{[16] $\vdash \Box no(p,w)\land most(p,n)\rightarrow not~all(n,w)$}%
		\hfill (i.e. $\Box$EMO-3, by [15] and Fact [2.4])\par
		
		\noindent
		\makebox[0pt][l]{[17] $\vdash \Box no(p,w)\land \Box most(p,n)\rightarrow not~all(n,w)$}\par
		\hfill (i.e. $\Box$E$\Box$MO-3, by [16], Fact [4.10] and Rule 1)\par
		
		\noindent
		\makebox[0pt][l]{[18] $\vdash \Box no(p,w)\land most(p,n)\rightarrow \SquareDiamond not~all(n,w)$}\par
		\hfill (i.e. $\Box$EM$\SquareDiamond$O-3, by [16], Fact [4.12] and Rule 2)\par
	
		\noindent
		\makebox[0pt][l]{[19] $\vdash \Box no(p,w)\land \Box most(p,n)\rightarrow \SquareDiamond not~all(n,w)$}\par
		\hfill (i.e. $\Box$E$\Box$M$\SquareDiamond$O-3, by [17], Fact [4.12] and Rule 2)\par
			
		\noindent
		\makebox[0pt][l]{[20] $\vdash \Box no(w,p)\land most(p,n)\rightarrow not~all(n,w)$}%
		\hfill (i.e. $\Box$EMO-4, by [16] and Fact [3.2])\par
		
		\noindent
		\makebox[0pt][l]{[21] $\vdash \Box no(w,p)\land \Box most(p,n)\rightarrow not~all(n,w)$}%
		\hfill (i.e. $\Box$E$\Box$MO-4, by [17] and Fact [3.2])\par
		
		\noindent
		\makebox[0pt][l]{[22] $\vdash \Box no(w,p)\land most(p,n)\rightarrow \SquareDiamond not~all(n,w)$}%
		\hfill (i.e. $\Box$EM$\SquareDiamond$O-4, by [18] and Fact [3.2])\par
		
		\noindent
		\makebox[0pt][l]{[23] $\vdash \Box no(w,p)\land \Box most(p,n)\rightarrow \SquareDiamond not~all(n,w)$}%
		\hfill (i.e. $\Box$E$\Box$M$\SquareDiamond$O-4, by [19] and Fact [3.2])\par
		
		\noindent
		\makebox[0pt][l]{[24] $\vdash \lnot some(p,w)\land most(p,n)\rightarrow \lnot \Box all(n,w)$}%
		\hfill (by [10] and Rule 3)\par
		
		\noindent
		\makebox[0pt][l]{[25] $\vdash no(p,w)\land most(p,n)\rightarrow \SquareDiamond \lnot all(n,w)$}%
		\hfill (by [24], Fact [2.4] and [5.1])\par
		
		\noindent
		\makebox[0pt][l]{[26] $\vdash no(p,w)\land most(p,n)\rightarrow \SquareDiamond not~all(n,w)$}%
		\hfill (i.e. EM$\SquareDiamond$O-3, by [25] and Fact [2.1])\par
		
		\noindent
		\makebox[0pt][l]{[27] $\vdash no(w,p)\land most(p,n)\rightarrow \SquareDiamond not~all(n,w)$}%
		\hfill (i.e. EM$\SquareDiamond$O-4, by [26] and Fact [3.2])\par
		
	\end{proof}
	
	Deriving valid classical syllogisms from the valid generalized syllogism AMI-1 can be approached as follows: (1) The classical syllogism AAI-1 is derived from AMI-1 by utilizing the antecedent strengthening rule. (2) Other valid classical syllogisms can be deduced from AAI-1 on the basis of deductive reasoning. The specific steps are shown in the following Theorem~\ref{thm:4.4}.
	
	\begin{thm}\label{thm:4.4}
		There are at least the following 8 valid classical syllogisms that can be mined from the generalized syllogism AMI-1:
	\end{thm}
	\begin{enumerate}[label={(\arabic*)}, leftmargin=3em, itemsep=3pt, topsep=4pt]
		\item $\vdash \text{AMI-1} \rightarrow \text{AAI-1};$
		\item $\vdash \text{AMI-1} \rightarrow \text{AAI-1} \rightarrow \text{AAI-4};$
		\item $\vdash \text{AMI-1} \rightarrow \text{AAI-1} \rightarrow \text{EAO-3};$
		\item $\vdash \text{AMI-1} \rightarrow \text{AAI-1} \rightarrow \text{EAO-3} \rightarrow \text{EAO-4};$
		\item $\vdash \text{AMI-1} \rightarrow \text{AAI-1} \rightarrow \text{AEO-2};$
		\item $\vdash \text{AMI-1} \rightarrow \text{AAI-1} \rightarrow \text{AEO-2} \rightarrow \text{AEO-4};$
		\item $\vdash \text{AMI-1} \rightarrow \text{AAI-1} \rightarrow \text{EAO-1};$
		\item $\vdash \text{AMI-1} \rightarrow \text{AAI-1} \rightarrow \text{EAO-1} \rightarrow \text{EAO-2}.$
	\end{enumerate}
		\begin{proof}\leavevmode\par
		\noindent
		\makebox[0pt][l]{[1] $\vdash all(n,w)\land most(p,n)\rightarrow some(p,w)$}%
		\hfill (i.e. AMI-1, Axiom A2)\par
		
		\noindent
		\makebox[0pt][l]{[2] $\vdash all(n,w)\land all(p,n)\rightarrow some(p,w)$}%
		\hfill (i.e. AAI-1, by [1], Fact [4.3] and Rule 1)\par
		
		\noindent
		\makebox[0pt][l]{[3] $\vdash all(n,w)\land all(p,n)\rightarrow some(w,p)$}%
		\hfill (i.e. AAI-4, by [2] and Fact [3.1])\par
		
		\noindent
		\makebox[0pt][l]{[4] $\vdash \lnot some(p,w)\land all(p,n)\rightarrow \lnot all(n,w)$}%
		\hfill (by [2] and Rule 3)\par
		
		\noindent
		\makebox[0pt][l]{[5] $\vdash no(p,w)\land all(p,n)\rightarrow not~all(n,w)$}%
		\hfill (i.e. EAO-3, by [4], Fact [2.1] and [2.4])\par
		
		\noindent
		\makebox[0pt][l]{[6] $\vdash no(w,p)\land all(p,n)\rightarrow not~all(n,w)$}%
		\hfill (i.e. EAO-4, by [5] and Fact [3.2])\par
		
		\noindent
		\makebox[0pt][l]{[7] $\vdash some(p,w)\land all(n,w)\rightarrow \lnot all(p,n)$}%
		\hfill (by [2] and Rule 3)\par
		
		\noindent
		\makebox[0pt][l]{[8] $\vdash no(p,w)\land all(n,w)\rightarrow not~all(p,n)$}%
		\hfill (i.e. AEO-2, by [7], Fact [2.1] and [2.4])\par
		
		\noindent
		\makebox[0pt][l]{[9] $\vdash no(w,p)\land all(n,w)\rightarrow not~all(p,n)$}%
		\hfill (i.e. AEO-4, by [8] and Fact [3.2])\par
		
		\noindent
		\makebox[0pt][l]{[10] $\vdash no\lnot(n,w)\land all(p,n)\rightarrow not~all\lnot(p,w)$}%
		\hfill (by [2], Fact [1.1] and [1.3])\par
		
		\noindent
		\makebox[0pt][l]{[11] $\vdash no(n,D{-}w)\land all(p,n)\rightarrow not~all(p,D{-}w)$}%
		\hfill (i.e. EAO-1, by [10] and Definition D3)\par
		
		\noindent
		\makebox[0pt][l]{[12] $\vdash no(D{-}w,n)\land all(p,n)\rightarrow not~all(p,D{-}w)$}%
		\hfill (i.e. EAO-2, by [11] and Fact [3.2])\par
		
	\end{proof}
	
	The following Theorem~\ref{thm:4.5} discusses the derivation of valid non-trivial classical modal syllogisms from the valid generalized syllogism AMI-1, which can be outlined as follows: (1) The classical syllogism AAI-1 can be deduced by applying the antecedent strengthening rule to AMI-1. (2) The classical modal syllogism AA$\SquareDiamond$I-1 can be derived from AAI-1 by the subsequent weakening rule. (3) Other valid classical modal syllogisms can be inferred from AA$\SquareDiamond$I-1 by means of knowledge reasoning.
	
	\begin{thm}\label{thm:4.5}
		There are at least the following 24 valid non-trivial classical modal syllogisms that can be mined from the generalized syllogism AMI-1:
	\end{thm}
	\begin{enumerate}[label={(\arabic*)}, leftmargin=3em, itemsep=3pt, topsep=4pt]
		\item $\vdash \text{AMI-1} \rightarrow \text{AAI-1} \rightarrow \text{AA$\SquareDiamond$I-1};$
		\item $\vdash \text{AMI-1} \rightarrow \text{AAI-1} \rightarrow \text{AA$\SquareDiamond$I-1} \rightarrow \Box\text{AA$\SquareDiamond$I-1};$
		\item $\vdash \text{AMI-1} \rightarrow \text{AAI-1} \rightarrow \text{AA$\SquareDiamond$I-1} \rightarrow \text{A}\Box\text{A$\SquareDiamond$I-1};$
		\item $\vdash \text{AMI-1} \rightarrow \text{AAI-1} \rightarrow \text{AA$\SquareDiamond$I-1} \rightarrow \Box\text{AA$\SquareDiamond$I-1} \rightarrow \Box\text{A}\Box\text{A$\SquareDiamond$I-1};$
		\item $\vdash \text{AMI-1} \rightarrow \text{AAI-1} \rightarrow \text{AA$\SquareDiamond$I-1} \rightarrow \text{AA$\SquareDiamond$I-4};$
		\item $\vdash \text{AMI-1} \rightarrow \text{AAI-1} \rightarrow \text{AA$\SquareDiamond$I-1} \rightarrow \Box\text{AA$\SquareDiamond$I-1} \rightarrow \text{A}\Box\text{A$\SquareDiamond$I-4};$
		\item $\vdash \text{AMI-1} \rightarrow \text{AAI-1} \rightarrow \text{AA$\SquareDiamond$I-1} \rightarrow \text{A}\Box\text{A$\SquareDiamond$I-1} \rightarrow \Box\text{AA$\SquareDiamond$I-4};$
		\item $\vdash \text{AMI-1} \rightarrow \text{AAI-1} \rightarrow \text{AA$\SquareDiamond$I-1} \rightarrow \Box\text{AA$\SquareDiamond$I-1} \rightarrow \text{A}\Box\text{A$\SquareDiamond$I-1} \rightarrow \Box\text{A}\Box\text{A$\SquareDiamond$I-4};$
		\item $\vdash \text{AMI-1} \rightarrow \text{AAI-1} \rightarrow \text{AA$\SquareDiamond$I-1} \rightarrow \Box\text{EAO-3};$
		\item $\vdash \text{AMI-1} \rightarrow \text{AAI-1} \rightarrow \text{AA$\SquareDiamond$I-1} \rightarrow \Box\text{EAO-3} \rightarrow \Box\text{E}\Box\text{AO-3}.$
		\item $\vdash \text{AMI-1} \rightarrow \text{AAI-1} \rightarrow \text{AA$\SquareDiamond$I-1} \rightarrow \Box\text{EAO-3} \rightarrow \Box\text{EAO-4};$
		\item $\vdash \text{AMI-1} \rightarrow \text{AAI-1} \rightarrow \text{AA$\SquareDiamond$I-1} \rightarrow \Box\text{EAO-3} \rightarrow \Box\text{E$\Box$AO-3} \rightarrow \Box\text{E$\Box$AO-4};$
		\item $\vdash \text{AMI-1} \rightarrow \text{AAI-1} \rightarrow \text{AA$\SquareDiamond$I-1} \rightarrow \Box\text{AA$\SquareDiamond$I-1} \rightarrow \Box\text{EA$\SquareDiamond$O-3};$
		\item $\vdash \text{AMI-1} \rightarrow \text{AAI-1} \rightarrow \text{AA$\SquareDiamond$I-1} \rightarrow \Box\text{AA$\SquareDiamond$I-1} \rightarrow \Box\text{EA$\SquareDiamond$O-3} \rightarrow \Box\text{E$\Box$A$\SquareDiamond$O-3};$
		\item $\vdash \text{AMI-1} \rightarrow \text{AAI-1} \rightarrow \text{AA$\SquareDiamond$I-1} \rightarrow \Box\text{AA$\SquareDiamond$I-1} \rightarrow \Box\text{EA$\SquareDiamond$O-3} \rightarrow \Box\text{EA$\Box$O-4};$
		\item $\vdash \text{AMI-1} \rightarrow \text{AAI-1} \rightarrow \text{AA$\SquareDiamond$I-1} \rightarrow \Box\text{AA$\SquareDiamond$I-1} \rightarrow \Box\text{EA$\SquareDiamond$O-3}  \rightarrow \Box\text{E$\Box$A$\SquareDiamond$O-3}\\ \rightarrow \Box\text{E$\Box$A$\SquareDiamond$O-4};$
		\item $\vdash \text{AMI-1} \rightarrow \text{AAI-1} \rightarrow \text{AA$\SquareDiamond$I-1} \rightarrow \text{EA$\SquareDiamond$O-1};$
		\item $\vdash \text{AMI-1} \rightarrow \text{AAI-1} \rightarrow \text{AA$\SquareDiamond$I-1} \rightarrow \text{EA$\SquareDiamond$O-1} \rightarrow \text{EA$\Box$O-1};$
		\item $\vdash \text{AMI-1} \rightarrow \text{AAI-1} \rightarrow \text{AA$\SquareDiamond$I-1} \rightarrow \text{EA$\SquareDiamond$O-1} \rightarrow \text{E$\Box$A$\SquareDiamond$O-1};$
		\item $\vdash \text{AMI-1} \rightarrow \text{AAI-1} \rightarrow \text{AA$\SquareDiamond$I-1} \rightarrow \text{EA$\SquareDiamond$O-1} \rightarrow \Box\text{EA$\SquareDiamond$O-1} \rightarrow \Box\text{E$\Box$A$\SquareDiamond$O-1}.$
		\item $\vdash \text{AMI-1} \rightarrow \text{AAI-1} \rightarrow \text{AA$\SquareDiamond$I-1} \rightarrow \text{EA$\SquareDiamond$O-1} \rightarrow \text{EA$\SquareDiamond$O-2};$
		\item $\vdash \text{AMI-1} \rightarrow \text{AAI-1} \rightarrow \text{AA$\SquareDiamond$I-1} \rightarrow \text{EA$\SquareDiamond$O-1} \rightarrow \Box\text{EA$\SquareDiamond$O-1} \rightarrow \Box\text{EA$\SquareDiamond$O-2};$
		\item $\vdash \text{AMI-1} \rightarrow \text{AAI-1} \rightarrow \text{AA$\SquareDiamond$I-1} \rightarrow \text{EA$\SquareDiamond$O-1} \rightarrow \text{E$\Box$A$\SquareDiamond$O-1} \rightarrow \text{E$\Box$A$\SquareDiamond$O-2};$
		\item $\vdash \text{AMI-1} \rightarrow \text{AAI-1} \rightarrow \text{AA$\SquareDiamond$I-1} \rightarrow \text{EA$\SquareDiamond$O-1} \rightarrow \Box\text{EA$\SquareDiamond$O-1}
		\rightarrow \Box\text{E$\Box$A$\SquareDiamond$O-1}\\ \rightarrow \Box\text{E$\Box$A$\SquareDiamond$O-2}.$
	\end{enumerate}
	\begin{proof}\leavevmode\par
		\noindent
		\makebox[0pt][l]{[1] $\vdash all(n,w)\land most(p,n)\rightarrow some(p,w)$}%
		\hfill (i.e. AMI-1, Axiom A2)\par
		
		\noindent
		\makebox[0pt][l]{[2] $\vdash all(n,w)\land all(p,n)\rightarrow some(p,w)$}%
		\hfill (i.e. AAI-1, by [1], Fact [4.3] and Rule 1)\par
		
		\noindent
		\makebox[0pt][l]{[3] $\vdash all(n,w)\land all(p,n)\rightarrow \SquareDiamond some(p,w)$}%
		\hfill (i.e. AA$\SquareDiamond$I-1, by [2], Fact [4.12] and Rule 2)\par
		
		\noindent
		\makebox[0pt][l]{[4] $\vdash \Box all(n,w)\land all(p,n)\rightarrow \SquareDiamond some(p,w)$}%
		\hfill (i.e. $\Box$AA$\SquareDiamond$I-1, by [3], Fact [4.10] and Rule 1)\par
		
		\noindent
		\makebox[0pt][l]{[5] $\vdash all(n,w)\land \Box all(p,n)\rightarrow \SquareDiamond some(p,w)$}%
		\hfill (i.e. A$\Box$A$\SquareDiamond$I-1, by [3], Fact [4.10] and Rule 1)\par
		
		\noindent
		\makebox[0pt][l]{[6] $\vdash \Box all(n,w)\land \Box all(p,n)\rightarrow \SquareDiamond some(p,w)$}\par
		
		\hfill (i.e. $\Box$A$\Box$A$\SquareDiamond$I-1, by [4], Fact [4.10] and Rule 1)\par
		
		\noindent
		\makebox[0pt][l]{[7] $\vdash all(n,w)\land all(p,n)\rightarrow \SquareDiamond some(w,p)$}%
		\hfill (i.e. AA$\SquareDiamond$I-4, by [3] and Fact [3.1])\par
		
		\noindent
		\makebox[0pt][l]{[8] $\vdash \Box all(n,w)\land all(p,n)\rightarrow \SquareDiamond some(w,p)$}%
		\hfill (i.e. A$\Box$A$\SquareDiamond$I-4, by [4] and Fact [3.1])\par
		
		\noindent
		\makebox[0pt][l]{[9] $\vdash all(n,w)\land \Box all(p,n)\rightarrow \SquareDiamond some(w,p)$}%
		\hfill (i.e. $\Box$AA$\SquareDiamond$I-4, by [5] and Fact [3.1])\par
		
		\noindent
		\makebox[0pt][l]{[10] $\vdash \Box all(n,w)\land \Box all(p,n)\rightarrow \SquareDiamond some(w,p)$}%
		\hfill (i.e. $\Box$A$\Box$A$\SquareDiamond$I-4, by [6] and Fact [3.1])\par
		
		\noindent
		\makebox[0pt][l]{[11] $\vdash \lnot \SquareDiamond some(p,w)\land all(p,n)\rightarrow \lnot all(n,w)$}%
		\hfill (by [3] and Rule 3)\par
		
		\noindent
		\makebox[0pt][l]{[12] $\vdash \Box \lnot some(p,w)\land all(p,n)\rightarrow not~all(n,w)$}%
		\hfill (by [11], Fact [2.1] and [5.2])\par
		
		\noindent
		\makebox[0pt][l]{[13] $\vdash \Box no(p,w)\land all(p,n)\rightarrow not~all(n,w)$}%
		\hfill (i.e. $\Box$EAO-3, by [12] and Fact [2.4])\par
		
		\noindent
		\makebox[0pt][l]{[14] $\vdash \Box no(p,w)\land \Box all(p,n)\rightarrow not~all(n,w)$}\par
		\hfill (i.e. $\Box$E$\Box$AO-3, by [13], Fact [4.10] and Rule 1)\par
		
		\noindent
		\makebox[0pt][l]{[15] $\vdash \Box no(w,p)\land all(p,n)\rightarrow not~all(n,w)$}%
		\hfill (i.e. $\Box$EAO-4, by [13] and Fact [3.2])\par
		
		\noindent
		\makebox[0pt][l]{[16] $\vdash \Box no(w,p)\land \Box all(p,n)\rightarrow not~all(n,w)$}%
		\hfill (i.e. $\Box$E$\Box$AO-4, by [14] and Fact [3.2])\par
		
		\noindent
		\makebox[0pt][l]{[17] $\vdash \lnot \SquareDiamond some(p,w)\land all(p,n)\rightarrow \lnot \Box all(n,w)$}%
		\hfill (by [4] and Rule 3)\par
		
		\noindent
		\makebox[0pt][l]{[18] $\vdash \Box \lnot some(p,w)\land all(p,n)\rightarrow \SquareDiamond \lnot all(n,w)$}%
		\hfill (by [17], Fact [5.1] and [5.2])\par
		
		\noindent
		\makebox[0pt][l]{[19] $\vdash \Box no(p,w)\land all(p,n)\rightarrow \SquareDiamond not~all(n,w)$}%
		\hfill (i.e. $\Box$EA$\SquareDiamond$O-3, by [18], Fact [2.1] and [2.4])\par
		
		\noindent
		\makebox[0pt][l]{[20] $\vdash \Box no(p,w)\land \Box all(p,n)\rightarrow \SquareDiamond not~all(n,w)$}\par
		\hfill (i.e. $\Box$E$\Box$A$\SquareDiamond$O-3, by [19], Fact [4.10] and Rule 1)\par
		
		\noindent
		\makebox[0pt][l]{[21] $\vdash \Box no(w,p)\land all(p,n)\rightarrow \SquareDiamond not~all(n,w)$}%
		\hfill (i.e. $\Box$EA$\SquareDiamond$O-4, by [19] and Fact [3.2])\par
		
		\noindent
		\makebox[0pt][l]{[22] $\vdash \Box no(w,p)\land \Box all(p,n)\rightarrow \SquareDiamond not~all(n,w)$}%
		\hfill (i.e. $\Box$E$\Box$A$\SquareDiamond$O-4, by [20] and Fact [3.2])\par
		
		\noindent
		\makebox[0pt][l]{[23] $\vdash no\lnot (n,w)\land all(p,n)\rightarrow \SquareDiamond not~all\lnot (p,w)$}%
		\hfill (by [3], Fact [1.1] and [1.3])\par
		
		\noindent
		\makebox[0pt][l]{[24] $\vdash no(n,D{-}w)\land all(p,n)\rightarrow \SquareDiamond not~all(p,D{-}w)$}\par
		\hfill (i.e. EA$\SquareDiamond$O-1, by [23] and Definition D3)\par
		
		\noindent
		\makebox[0pt][l]{[25] $\vdash \Box no(n,D{-}w)\land all(p,n)\rightarrow \SquareDiamond not~all(p,D{-}w)$}\par
		\hfill (i.e. $\Box$EA$\SquareDiamond$O-1, by [24], Fact [4.10] and Rule 1)\par
		
		\noindent
		\makebox[0pt][l]{[26] $\vdash no(n,D{-}w)\land \Box all(p,n)\rightarrow \SquareDiamond not~all(p,D{-}w)$}\par
		\hfill (i.e. EA$\Box\SquareDiamond$O-1, by [24], Fact [4.10] and Rule 1)\par
		
		\noindent
		\makebox[0pt][l]{[27] $\vdash \Box no(n,D{-}w)\land \Box all(p,n)\rightarrow \SquareDiamond not~all(p,D{-}w)$}\par
		\hfill (i.e. $\Box$E$\Box$A$\SquareDiamond$O-1, by [25], Fact [4.10] and Rule 1)\par
		
		\noindent
		\makebox[0pt][l]{[28] $\vdash no(D{-}w,n)\land all(p,n)\rightarrow \SquareDiamond not~all(p,D{-}w)$}
		\hfill (i.e. EA$\SquareDiamond$O-2, by [24] and Fact [3.2])\par
		
		\noindent
		\makebox[0pt][l]{[29] $\vdash \Box no(D{-}w,n)\land all(p,n)\rightarrow \SquareDiamond not~all(p,D{-}w)$}\par
		\hfill (i.e. $\Box$EA$\SquareDiamond$O-2, by [25] and Fact [3.2])\par
		
		\noindent
		\makebox[0pt][l]{[30] $\vdash no(D{-}w,n)\land \Box all(p,n)\rightarrow \SquareDiamond not~all(p,D{-}w)$}\par
		\hfill (i.e. E$\Box$A$\SquareDiamond$O-2, by [26] and Fact [3.2])\par
		
		\noindent
		\makebox[0pt][l]{[31] $\vdash \Box no(D{-}w,n)\land \Box all(p,n)\rightarrow \SquareDiamond not~all(p,D{-}w)$}\par
		\hfill (i.e. $\Box$E$\Box$A$\SquareDiamond$O-2, by [27] and Fact [3.2])\par
		
	\end{proof}
	
	So far, from the perspective of knowledge reasoning, this paper has discussed in detail how to deduce the above four types of syllogisms based on the generalized syllogism AMI-1.
	
	\section{Discourse Reasoning}
	
	This part analyzes how discourse reasoning is constructed by the nesting of the following four types of syllogisms: classical ones, classical modal ones, generalized ones, and generalized modal ones. Now, let’s consider the following examples of discourse reasoning in natural language:
	
	‘(1) All people who have systematically studied educational theory have mastered basic teaching methods. (2) All graduates from normal universities have systematically studied educational theory. (3) Therefore, all graduates from normal universities have mastered basic teaching methods. (4) Most people who have mastered basic teaching methods can effectively manage classes. (5) Therefore, most graduates from normal universities can effectively manage classes. (6) Most people who can effectively manage classes possibly won teaching innovation awards. (7) Therefore, some graduates from normal universities possibly won teaching innovation awards. (8) All people who have won teaching innovation awards can improve students’ grades. (9) Therefore, some graduates from normal universities can possibly improve students’ grades.’
	
	This discourse reasoning is nested by the following four different types of syllogisms:

	\begin{description}[style=unboxed,leftmargin=0pt]
		\item[\textbf{[1]The classical syllogism AAA-1}] \leavevmode\par
		\noindent Major premise: (1) All people who have systematically studied educational theory have mastered basic teaching methods.\\
		\noindent Minor premise: \uline{(2) All graduates from normal universities have systematically studied educational theory.}\\
		\noindent Conclusion: (3) All graduates from normal universities have mastered basic teaching methods.
	\end{description}
	
	Let $c$ be a variable of graduates from normal universities, 
	$z$ be a variable of people who have mastered basic teaching methods, 
	and $k$ be a variable of people who have systematically studied educational theory. 
	The above classical syllogism is symbolized as 
	$all(k, z) \land all(c, k) \rightarrow all(c, z)$, 
	which is valid.
	\begin{proof}	
		Assume that both premises are true. According to Definition D5, 
		$all(k, z)$ and $all(c, k)$ are true when and only when 
		$K \subseteq Z$ and $C \subseteq K$ are true in any real world. 
		Therefore, $C \subseteq Z$ is also true, and in the light of Definition~D5, 
		$all(c, z)$ is true. In other words, the classical syllogism AAA-1 is valid.
	\end{proof}

	\begin{description}[style=unboxed,leftmargin=0pt]
		\item[\textbf{[2]The generalized syllogism MAM-1}] \leavevmode\par
		\noindent Major premise: (4) Most people who have mastered basic teaching methods can effectively manage classes.\\
		\noindent Minor premise: \uline{(3) All graduates from normal universities have mastered basic teaching methods}.\\
		\noindent Conclusion: (5) Most graduates from normal universities can effectively manage classes.
	\end{description}
	
	Let $p$ be a variable representing people who can effectively manage classes. The above generalized syllogism is symbolized as $most(z,p) \land all(c,z) \rightarrow most(c,p)$. It is a valid syllogism.
	
	\begin{proof}
		Assume that both premises are true, in line with Definition D5 and D7, 
		$most(z, p)$ and $all(c, z)$ are true when and only when 
		$|Z \cap P| > 0.5|Z|$ is true in any real world, 
		and $C \subseteq Z$ is true in any real world, respectively. 
		Consequently, $|C \cap P| > 0.5|C|$ is true in any real world. 
		According to Definition D7, $most(c, p)$ is true when and only when 
		$|C \cap P| > 0.5|C|$ is true in any real world, 
		thus $most(c, p)$ is true. 
		It can be concluded that the generalized syllogism MAM-1 is valid.
	\end{proof}
	
	\begin{description}[style=unboxed,leftmargin=0pt]
		\item[\textbf{[3]The generalized modal syllogism $\SquareDiamond \text{MM}\SquareDiamond\text{I-1}$}] \leavevmode\par
		\noindent Major premise: (6) Most people who can effectively manage classes possibly won teaching innovation awards.\\
		\noindent Minor premise: \uline{(5) Most graduates from normal universities can effectively manage classes}.\\
		\noindent Conclusion: (7) Some graduates from normal universities possibly won teaching innovation awards.
	\end{description}
	
	Let $r$ be a variable representing people who won teaching innovation awards. 
	The above generalized modal syllogism can thus be formalized as 
	$\SquareDiamond most(p, r) \land most(c, p) \rightarrow \SquareDiamond some(c, r)$. 
	Its validity is demonstrated as follows:
	
	\begin{proof}
		Assume that both premises are true. According to Definitions D7 and D9, 
		$\SquareDiamond most(p, r)$ is true when and only when $|P \cap R| > 0.5|P|$ 
		is true in at least one possible world, and $most(c, p)$ is true when and only when 
		$|C \cap P| > 0.5|C|$ is true in any real world. Since any real world is also a possible world, 
		thus $|P \cap R| > 0.5|P|$ and $|C \cap P| > 0.5|C|$ are true in at least one possible world. 
		Therefore, $C \cap R \neq \emptyset$ is true in at least one possible world. 
		Then, according to Definition D8, $\SquareDiamond some(c, r)$ is true. 
		Thus, the generalized modal syllogism $\SquareDiamond \text{MM}\SquareDiamond\text{I-1}$ is valid.
	\end{proof}
	
	\begin{description}[style=unboxed,leftmargin=0pt]
		\item[\textbf{[4]The classical modal syllogism $\text{A}\SquareDiamond \text{I}\SquareDiamond\text{I-1}$}] \leavevmode\par
		\noindent Major premise: (8) All people who have won teaching innovation awards can improve students’ grades.\\
		\noindent Minor premise: \uline{(7) Some graduates from normal universities possibly won teaching innovation awards}.\\
		\noindent Conclusion: (9) Some graduates from normal universities can possibly improve students’ grades.
	\end{description}
	
	Let $n$ be a variable representing people who can improve students’ grades. 
	The above classical modal syllogism can be symbolized as 
	$all(r, n) \land \SquareDiamond some(c, r) \rightarrow \SquareDiamond some(c, n)$. 
	Regarding the validity of this syllogism, the proof is presented below.
	\begin{proof}
		Assume that both premises are true. In the light of Definition~D5, 
		$all(r,n)$ is true when and only when $R \subseteq N$ is true in any real world. 
		Since any real world is also a possible world, $R \subseteq N$ is true in at least one possible world. 
		According to Definition~D8, $\SquareDiamond some(c,r)$ is true when and only when $C \cap R \neq \emptyset$ 
		is true in at least one possible world. Consequently, $R \subseteq N$ and $C \cap R \neq \emptyset$ 
		are true in at least one possible world, and then $C \cap N \neq \emptyset$ is true in at least one possible world. 
		According to Definition~D8, $\SquareDiamond some(c,n)$ is true. 
		That means that the classical modal syllogism $\text{A}\SquareDiamond \text{I}\SquareDiamond\text{I-1}$ is valid.
	\end{proof}
	
	Therefore, the discourse reasoning composed of the above four valid syllogisms is valid. This discourse reasoning can be represented by the following tree diagram:
	\[
	\begin{array}{cccccc}
		&\hspace{10em} \underline{all(k, z) \hspace{3em} all(c, k)} & \\[6pt]
		&\hspace{7em}\underline{most(z, p) \hspace{3em} all(c, z)} & \\[6pt]
		&\hspace{4em}\underline{\SquareDiamond most(p, r) \hspace{3em} most(c, p)} & \\[6pt]
		&\underline{all(r, n) \hspace{3em} \SquareDiamond some(c, r)} & \\[6pt]
		& \SquareDiamond some(c, n) &
	\end{array}
	\]
	
	In this tree diagram, the conclusion of the classical syllogism $\text{AAA-1}$ serves as the minor premise of the generalized syllogism $\text{MAM-1}$, the conclusion of the syllogism $\text{MAM-1}$ acts as the minor premise of the generalized modal syllogism $\SquareDiamond \text{MM} \SquareDiamond \text{I-1}$, and the conclusion of the  syllogism $\SquareDiamond \text{MM} \SquareDiamond \text{I-1}$
	 functions as the minor premise of the classical modal syllogism $\text{A}\SquareDiamond\text{I}\SquareDiamond\text{I-1}$
	.
	
	In conclusion, the validity of this nested discourse reasoning is proved by the step-by-step proof of the validity of its constituent syllogisms. The discourse reasoning nested by several syllogisms is valid, if and only if, all syllogisms involved are valid. Otherwise, even if one syllogism is invalid, the entire discourse reasoning is invalid.
	\section{Conclusion and Future Work}
	
	This paper presents an analytical framework for complex syllogistic reasoning that incorporates quantifier-driven mechanisms in Square\{\textit{all}\} and Square\{\textit{most}\}. To achieve this, the paper initially formalizes generalized syllogisms, then proves the validity of the syllogism AMI-1 with the generalized quantifier most, and further deduces the other (19+22+8+24=)73 valid syllogisms. The mutual reducible relationships of all these valid syllogisms are for two reasons: (1) within Square\{\textit{all}\}, any of the four classical quantifiers (that is, \textit{no, all, not all, some}) can be used to define the other three; (2) similarly, so can any of the four generalized quantifiers in Square\{\textit{most}\} (that is, \textit{fewer than half of the, most, at least half of the, at most half of the}).
	
	In this paper, the 73 valid syllogisms are deduced from the valid generalized syllogism $\text{AMI-1}$ based on the given rules, facts and definitions. However, it is possible that not all derivable syllogisms have been included, suggesting that more syllogisms may be derivable from the syllogism $\text{AMI-1}$ by making full use of these rules, facts and definitions. Moreover, does the generalized syllogism $\text{AMI-1}$ in this paper have other meta-logical properties including soundness and completeness? This issue calls for a thorough analysis of the topic.

	\bibliographystyle{alphaurl}
	\bibliography{references}
	
\end{document}